\documentclass[letterpaper, 10 pt, conference]{ieeeconf}  

\IEEEoverridecommandlockouts                              
\usepackage[utf8]{inputenc}
\usepackage[T1]{fontenc}
\usepackage{amsmath}

\usepackage{amsthm}
\usepackage{graphicx}
\usepackage{subfig}
\usepackage[colorlinks=true, allcolors=blue]{hyperref}

\usepackage{amssymb}
\usepackage{caption}
\usepackage{titlesec}
\usepackage{thmtools}
\usepackage{mathptmx}

\newtheoremstyle{assumption}
  {0 pt}  
  {0 pt}  
  {\rmfamily\itshape} 
  {}  
  {\normalfont\bfseries} 
  {.}  
  { } 
  {}  
\theoremstyle{assumption}
\newtheorem{assumption}{Assumption}
\newtheorem{lemma}{Lemma}
\newtheorem{property}{Property}
\newtheorem{theorem}{Theorem}
\newtheorem{definition}{Definition}
\newcommand{\revise}[1]{\textcolor{black}{#1}}
\newcommand{\new}[1]{\textcolor{black}{#1}}

\titlespacing{\subsection}{0pt}{2pt plus 0.1pt minus 0.1pt}{1pt plus 0.1pt minus 0.1pt}
\allowdisplaybreaks
\title{\LARGE \bf
Singularity-Avoidance Control of Robotic Systems with Model Mismatch and Actuator Constraints
}

\author{Mingkun Wu, Alisa Rupenyan, \textit{Member, IEEE} and Burkhard Corves
\thanks{This work was supported in part by China Scholarship Council under Grant 202106250025. AR acknowledges funding support by the Joh. Jacob Rieter-Stiftung (\textit{Corresponding author:  Mingkun Wu})}
\thanks{Mingkun Wu and Burkhard Corves are with the Institute of Mechanism Theory, Machine Dynamics and Robotics, RWTH Aachen University, Aachen 52062, Germany. 
(e-mail: wu@igmr.rwth-aachen.de, corves@igmr.rwth-aachen.de)}
\thanks{Alisa Rupenyan is with ZHAW Centre for Artificial Intelligence, ZHAW Zürich University
of Applied Sciences, Winterthur 8401, Switzerland (e-mail: rupn@zhaw.ch)}}

\begin{document}

\maketitle
\thispagestyle{empty}
\pagestyle{empty}

\begin{abstract}
\revise{Singularities, manifesting as special configuration states, deteriorate robot performance and may even lead to a loss of control over the system. This paper addresses the kinematic singularity concerns in robotic systems with model mismatch and actuator constraints through control barrier functions (CBFs).} We propose a learning-based control strategy to prevent robots entering singularity regions. More precisely, we leverage Gaussian process (GP) regression to learn the unknown model mismatch, where the prediction error is restricted by a deterministic bound. Moreover, we offer the criteria for parameter selection to ensure the feasibility of CBFs subject to actuator constraints. The proposed approach is validated by high-fidelity simulations on a 2 degrees-of-freedom (DoFs) planar robot.

\end{abstract}

\section{INTRODUCTION}

Robots are becoming increasingly prevalent across various industries, such as robotic arms used in industrial production and parallel-mechanism based legged robots. Singularities, arising from specific geometric \revise{relationships} between links, can cause robots to lose or gain one or more DoFs, potentially leading to a loss of control over the system. Therefore, avoiding singular configurations is crucial to ensure safe operation for robotic systems. Directly modifying the reference trajectories is an effective method to avoid singularities. \revise{For example, a linear weighting method based post-processing non-singular trajectory generation method was proposed for a 5-DoFs hybrid machining robot in \cite{li2023effective}.} An algorithm based on output twist screws was presented in \cite{pulloquinga2023type} to address type II singularity in parallel mechanisms by modifying trajectories. However, the trajectory modification method lacks sufficient flexibility, as one has to repeat the process for different trajectories. Moreover, non-singular reference trajectories are unable to guarantee robots not entering singularity regions, due to the presence of control errors. 

In recent years, advancements in \new{safe optimization enabled by} CBFs offer a promising alternative solution to the singularity avoidance problem. CBFs are powerful tools for handling various constraints, which enable their application in numerous safety-critical fields \cite{wabersich2023data}. For example, one can leverage multiple CBFs to coordinate connected and automated agents at intersections \cite{katriniok2022control}, where the collision avoidance CBFs and velocity CBFs have to be jointly feasible under input constraints. A CBFs design methodology is proposed in \cite{cortez2020correct} for Euler-Lagrange systems with position, velocity and input constraints. CBFs can be also used to ensure the safety of learned models for control in robotic systems \cite{xiao2023barriernet}, and to impose safety-critical constraints in a continuous-time trajectory generation process \cite{sforni2024receding}. In order to handle model mismatch, robust CBFs are developed by introducing a compensation term based on the bound of uncertainty \cite{nguyen2021robust}. \revise{GP-based learning methods are another suitable method to tackle model mismatch \cite{jagtap2020control}, as they provide a quantification of the prediction uncertainty, which could be used to obtain the corresponding bound \cite{fisac2018general, balta2021learning}}.

There is little research regarding CBFs in addressing singularity concerns. \revise{In \cite{kurtz2021control}, CBFs were utilized} to tackle singularity problem in passivity-based control. However, the feasibility of CBFs in \cite{kurtz2021control} is based on the assumption that joints can provide unbounded torques, \revise{which does not precisely correspond to} the capabilities of motors in practice. In addition, model mismatch has not been addressed in \cite{kurtz2021control}.

This paper proposes a methodology for CBFs construction to address singularity avoidance problem in robotic systems with model mismatch and subject to actuator constraints. The primary contributions of this work are summarized as follows: (i) the theoretical guarantee of the feasibility of CBFs with model mismatch and actuator constraints is obtained, as well as the parameter selection criteria is provided, (ii) the model mismatch is learned using GP regression combined with a deterministic error bound, and (iii) the proposed approach is validated by high-fidelity 2 DoFs planar robot simulations on Simscape.

\revise{\textit{Notation:} $\mathbb{R}$ and $\mathbb{R}_{\ge0}$ denote the set of real, non-negative real numbers, respectively. The Euclidean norm is denoted by $\left\| \cdot \right\|$. $\mathbb{N}_n$ denoting the set of natural numbers $\{1,\cdots, n\}$. The matrix inequality $A\le B$ for matrices A and B means that the matrix $B-A$ is positive semidefinite. $e_i$ denotes the $i$th column of $n$th-order identity matrix $I_n$.}

\section{Preliminaries}
In this section, we recall some basic concepts about CBFs.
Consider a nonlinear affine system as follows:
\begin{equation}
    \dot x = f(x)+g(x)u,
\end{equation}
where $x\in\mathbb{R}^n$ and $u\in\mathcal{U}\subset\mathbb{R}^m$ denote the system state and control input, respectively. The functions $f(x)\ :\ \mathbb{R}^n\to\mathbb{R}^n$ and $g(x)\ :\ \mathbb{R}^m\to\mathbb{R}^n$ are assumed to be locally Lipschitz. Moreover, the system is forward complete.

Let $h(x)\ :\ \mathbb{R}^n\to\mathbb{R}$ be a continuously differentiable function related to safety concerns, then, the closed set $\mathcal{C}$ associated to $h(x)$ is defined by:
\begin{align}
    \mathcal{C} :=\{x\in\mathbb{R}^n: h(x)\ge0\}. \label{set C}
\end{align}

If for any initial state $x(t_0)\in\mathcal{C}$, $x(t)\in\mathcal{C}$ for all $t\in\mathbb{R}_{\ge0}$, then $\mathcal{C}$ is forward invariant and the constraint satisfaction is ensured.

We also present the definition of extended class-$\mathcal{K}$ function, \revise{which will be utilized in the definition of CBFs,} as follows:
\begin{definition}\label{class k}
    A continuous function $\alpha: \mathbb{R}\to\mathbb{R}$ is an extended class-$\mathcal{K}$ function if it is strictly increasing and with $\alpha(0) = 0$.
\end{definition}
With the assistance of Definition 1, the definition of CBFs is given as:
\begin{definition}
    Given a set $\mathcal{C}$ defined by \eqref{set C}, $h(x)$ is a CBF if there exists an extend class-$\mathcal{K}$ function, and for  such that:
    \begin{align}
        \sup_{u\in\mathcal{U}}[L_fh(x)+L_gh(x)u+\alpha(h(x))]\ge0,
    \end{align}
    where $L_fh(x)$ and $L_gh(x)$ denote the Lie derivative of $h(x)$ with respect to $x$ and are given by $L_fh(x)=\frac{\partial h(x)}{\partial x}f(x)$ and $L_gh(x)=\frac{\partial h(x)}{\partial x}g(x)$.
\end{definition}

\section{Problem Formulation}
Consider the following uncertain robotic system
\begin{align}
    \dot q=&v\nonumber\\
    \dot v =&M(q)^{-1}(u-C(q,v)v-G(q)-d(x))\label{system},
\end{align}
where $x=[q,\ v]^T\in\mathbb{R}^{2n}$ with $q=[q_1,\cdots,q_n]^T,\ v=[v_1,\cdots,v_n]^T\in\mathbb{R}^n$ denote the system states (i.e., angular positions and velocities). $M(q) \text{ and } C(q, v)\in\mathbb{R}^{n\times n}$, $ u=[u_1,\ \cdots,\ u_m]^T\text{ and }G(q)\in\mathbb{R}^n$ denote inertia matrix, centrifugal force matrix, control input and gravity vector, respectively. $d(x)=[d_1(x),\ \cdots,\ d_n(x)]\in\mathbb{R}^n$ denotes model mismatch. \revise{In order to facilitate the establishment of prediction error bounds in the following sections,} we assume that $d(x)$ satisfies the following assumption.

\begin{assumption}[\cite{hashimoto2022learning}]\label{RKHS}
    For a given kernel function $k_i$, $\mathcal{H}_i$ is the reproducing kernel Hilbert space (RKHS) corresponding to $k_i$ with the induced norm denoted by $\left\| \cdot \right\|_{k_i}$. The unknown function $d_i(x)$ belongs to $\mathcal{H}_i$ for all $i=1,\ \cdots,\ n$, then, \revise{its RKHS norm is bounded by a well-defined known constant $B_i\in\mathbb{R}_{\ge0}$, i.e., $\left\| d_i \right\|_{k_i}\le B_i$.}
\end{assumption}
Assumption 1 defines a potential function space of the unknown function $d_i(x)$. By choosing universal kernels, $\mathcal{H}_i$ contains all continuous functions, ensuring the property of universal approximation. As $B_i$ can be obtained by some data-driven methods \cite{scharnhorst2022robust}, Assumption 1 is not strict in practice.

To facilitate the feasibility analysis of CBFs, we also make an assumption for system state $q$ as follows:
\begin{assumption} \label{q constraint}
    The system state $q$ is bounded by hard constraints, i.e., $q\in\mathcal{Q}\subset\mathbb{R}^n$, where $\mathcal{Q}:=[q_{min},\ q_{max}]$ and $q_{max}=-q_{min}$.
\end{assumption}
Assumption 2 is reasonable in practice, since there must be restrictions of the rotation of joints for real robots, especially for industrial robots. We also assume that system \eqref{system} satisfies the following well-known properties.
\begin{property}\label{property1}
    M(q) is a symmetric, positive definite matrix, which satisfies $m_{min}I_n<M(q)^{-1}<m_{max}I_n$.
\end{property}
\begin{property}\label{property2}
    There exist $c_{max}\ \text{and}\ g_{max}\in \mathbb{R}_{>0}$ such that $C(q,v)\le c_{max}\left\|v\right\|$ and $\left\|G(q)\right\|\le g_{max}$.
\end{property}
Indeed, Assumption 2 also implies that Properties 1 and 2 holds, due to the smoothness of $M,\ C\ \text{and}\ G$.
\begin{figure}[!t]
  \centering
  \subfloat[]{\includegraphics[width=0.5\linewidth]{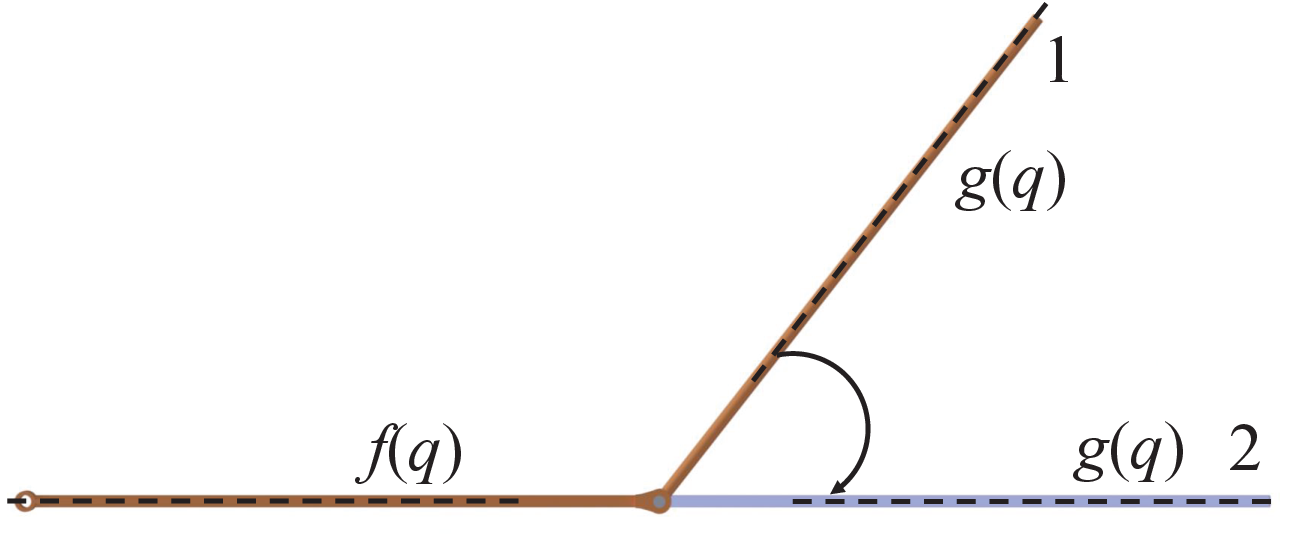}}
  \subfloat[]{\includegraphics[width=0.5\linewidth]{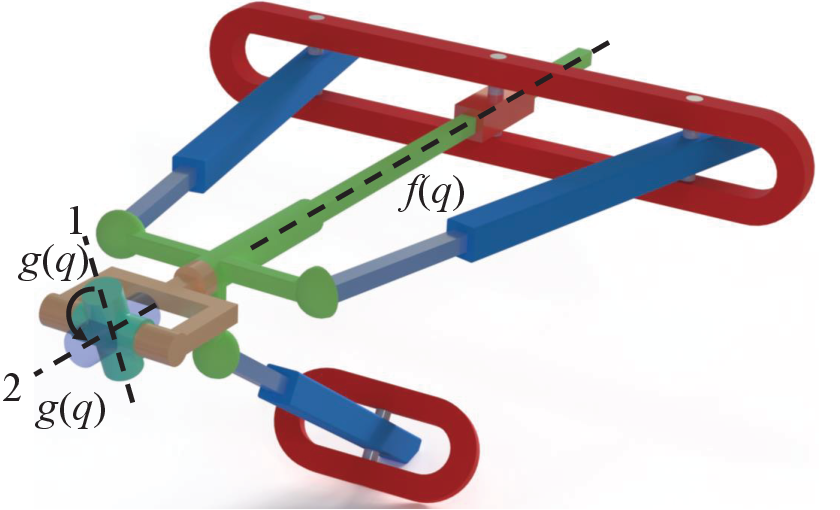}}
  \caption{Singularities occur when these two robots are in configuration 2. (a) 2 DoFs manipulator. (b) 5 DoFs robot.}
  \label{robots}
\end{figure}
The aim of this paper is to prevent robots from entering singularity regions. \revise{As shown in Fig. \ref{robots}}, one primary type of singularities occurs when the direction of two links $f(q)$ and $g(q)$ are parallel \cite{li2023effective, liu2012new}. Accordingly, we define the following configuration constraint.
\begin{equation}
    \arccos(f(q)^Tg(q))=0,
\end{equation}
where $f(q)$ and $g(q) \in \mathbb{R}^3$ denote two unit direction vectors. Adopting the singular cone concept \cite{lin2016improving} \new{that defines the singular domain by the angle between $f(q)$ and $g(q)$}, the following relationships should be satisfied to avoid singularities
\begin{align}
    z(q):=1-\epsilon-f(q)^Tg(q)\ge0\ \label{hard constraints},
\end{align}
where $\epsilon\in\mathcal{A}\subset\mathbb{R}_{>0}$ denotes a safety threshold determined by the half-angle of the singular 
cone. \eqref{hard constraints} indicates that we can established a singularity constraint as follows: 
\begin{align}
    \mathcal{Z} := \{q\in\mathbb{R}^n:z(q)\ge0\}. \label{configuration constraint}
\end{align}

In addition to the aforementioned singularity constraints, we also consider the following velocity constraints, which are important not only for safety concerns but also for the feasibility analysis of singularity CBFs.
\begin{align}
    \mathcal{V}_i:=\{v_i\in\mathbb{R}:\overline{b}_i(v_i)\ge0,\ \underline{b}_i(v_i)\ge0\},\ \forall i\in\mathbb{N}_n  \label{velocity constraint}
\end{align}
where \revise{$\overline{b}_i(v_i)$ and $\underline{b}_i(v_i)$ are defined as:}
\begin{align}
    \overline{b}_i(v_i):=v_{max}-v_i,\ \underline{b}_i(v_i):=v_i-v_{min} \label{func:b},
\end{align}
where for simplicity, we assume $v_{max}=-v_{min}$.

Due to physical constraints, robotic systems can only provide limited actuation torques. Thus, we define the following actuator constraints.
\begin{align}
    \mathcal{U}:=\{u_i\in \mathbb{R}, \forall i\in\mathbb{N}_n:u_{max}-u_i\ge0,\ u_i-u_{min}\ge0\}. \label{input constraint}
\end{align}
where similarly we assume $u_{max}=-u_{min}$ for all $i\in\mathbb{N}_n$.
\section{Control Barrier Function approach for singularity avoidance}
In order to ensure the set $\mathcal{Z}$ is forward invariant, $\dot{z}(q)\geq -\alpha(z(q))$ must hold where $\alpha$ is an extended class-$\mathcal{K}$ function. The derivative of $z(q)$ with respect to time is obtained as $\dot{z}(q) = -\left(\frac{\partial\eta(q) }{\partial q}\right)^Tv$ where $\eta(q) = f(q)^Tg(q)$. The relative degree of the system is 2, which means that no control input can directly ensure the above condition holds. 

\new{To enable the system input to directly act on the safety constraints}, we construct a new constraint as follows:
\begin{align}
    h(x) = \dot{z}(q)+\gamma\beta_1(z(q)),\label{func:h}
\end{align}
where $\gamma\in\mathbb{R}_{>0}$ is a user-design parameter and $\beta_1$ is an extended class-$\mathcal{K}$ function.

we define the closed set associated to \eqref{func:h} as follows:
\begin{align}
    \mathcal{C} :=\{x\in\mathbb{R}^{2n}:h(x)\ge0\},\label{set:c}
\end{align}

\revise{The objective of this paper is to ensure that robots operate safely: specifically, without violating singularity or velocity constraints, by using a safety filter based on CBFs. We can now formulate this objective as the following optimization problem:
\begin{align}
    \min_{u\in\mathcal{U}}\ &\left\|u-u_{nom}\right\|^2\label{optimization}\\
    \text{s.t.}\ &\dot{h}(x)\ge-\delta\beta_2(h(x))\label{dhx}\\
    &\dot{\overline{b}}_i(v_i)\ge-k\beta_3(\overline{b}_i(v_i))\\ &\dot{\underline{b}}_i(v_i)\ge-k\beta_3(\underline{b}_i(v_i))\label{dbx}\\
    &u_i\in\mathcal{U},\forall i\in\mathbb{N}_n,
\end{align}
where $q\in \mathcal{Z}\cap\mathcal{C}$ and $v_i\in \mathcal{C}\cap\mathcal{V}_i,\ \forall i\in\mathbb{N}_n$. $\beta_2$ and $\beta_3$ are extended class-$\mathcal{K}$ functions. $\delta\in\mathbb{R}_{>0}$ and $k\in\mathbb{R}_{>0}$ is two user-design parameters. $u_{nom}\in\mathcal{U}$ denotes a  nominal controller, such as PID and cascade controllers \cite{khosravi2022safety}.}

\revise{Now, we are aiming to find conditions that ensure \eqref{dhx} - \eqref{dbx} hold when the system \eqref{system} is subject to model mismatch and actuator constraints.}

\subsection{Singularity constraints}
\revise{In this section, we deduce a sufficient condition that can ensure \eqref{dhx} holds, and provide a criterion for parameter selection.}
By substituting the formulation of $\dot h(x)$ into \eqref{dhx}, and combining it with \eqref{system}, we have

%


\begin{align}
    -\Gamma^TM(q)^{-1}(u-C(q,v)v-G(q)-d(x))\nonumber\\-v^T\left(\frac{\partial^2\eta }{\partial q^2}\right)^Tv-\gamma\frac{\partial \beta_1}{\partial z}\Gamma^Tv\ge-\delta\beta_2(h(x)).\label{deri2}
\end{align}
where $\Gamma = \frac{\partial \eta}{\partial q}$.
Obviously, due to the mismatch term $d(x)$, it is difficult to ensure \eqref{deri2} holds. In this paper, we leverage GP regression to learn $d(x)$. More precisely, let $\mathcal{GP}(0,k_i)$ be the GP prior for $d_i(x)$ with zero mean, where $k_i:\ \mathbb{X}\times\mathbb{X}\to\ \mathbb{R}$ denotes kernel functions. Given a dataset $\mathcal{D}:=\{(x_i,Y_i)\}_{i=1}^M$ with $M$ data points in which $Y^i:=u^i-M(q^i)\ddot q^i-C(q^i,v^i)v^i-G(q^i)\in\mathbb{R}^n$, the prediction of $d(x)$ is characterized by the mean $\mu(x) = [\mu_1(x),\ \cdots,\ \mu_n(x)]^T$ and variance $\sigma^2(x)=[\sigma_1^2(x),\ \cdots,\ \sigma_n^2(x)]^T$ as:
\begin{align}
    \mu_i(x)&:=k_{\mathcal{D},i}^T(K_{\mathcal{D},i}+\sigma^2_vI_M)^{-1}y_{\mathcal{D},i},\\
    \sigma_i^2(x)&:=k_i(x,x)-k_{\mathcal{D},i}^T(K_{\mathcal{D},i}+\sigma^2_vI_M)^{-1}k_{\mathcal{D},i},
\end{align}
where $y_{\mathcal{D},i}:=[Y_1^i,\ \cdots,\ Y_M^i]^T$ for $i\in\mathbb{N}_n$ with $Y_j^i$ denoting the $i$th element of $Y_j$. $k_{\mathcal{D},i}=[k_i(x,x_1),\ \cdots,\ k_j(x,x_M)]^T$, and $K_{\mathcal{D},i}$ is defined as $\left[k_i(x^i,\ x^j)\right]_{i,j=1}^M$ for all $i\in\mathbb{N}_n$. $\sigma^2_v$ denotes the variance of noise. $I_M$ denotes $M$th order identity matrix.
Based on Assumption 1, the following lemma for prediction uncertainty of GP regression holds:
\begin{lemma}[\cite{hashimoto2022learning}]
    Suppose that Assumption 1 holds, and a training data $\mathcal{D}:=\{(x_i,y_i)\}_{i=1}^M$ is given. Then, for all $x\in\mathcal{X}$, the prediction error of GP regression is bounded by
    \begin{align}
        \left\|\mu(x)-d(x) \right\|\le \lambda(x):=\sqrt{\sum_{i=1}^{n}(B_i^2-\omega_i+M)\sigma_i^2}, \label{prediction error}
    \end{align}
    where $\omega_i=y_{\mathcal{D},i}^T\left( K_{\mathcal{D},i}+\sigma^2_vI_M \right)^{-1}y_{\mathcal{D},i}$.
\end{lemma}
Since $K_{D,i}+\sigma^2_vI_M$ is positive definite, we can also derive the following state-independent bound
\begin{align}
    \lambda(x)\le\bar\lambda:=\sqrt{\sum_{i=1}^{n}(B_i^2-\omega_i+M)\max_{\substack{
 v_i\in \mathcal{C}\cap\mathcal{V}_i\forall i\in\mathbb{N}_n\\
q\in \mathcal{Z}\cap\mathcal{C}
}}k_i(x,x)}. \label{statistic bound}
\end{align}

As a data-driven method, the decreasing of prediction error bound $\lambda$ in GP regression with the increasing of dataset size $M$ has been proven in \cite{hashimoto2022learning}. However, the increasing of $M$ also implies the increasing of computational time. Therefore, it is imperative to elaborately choose dataset to ensure a trade-off between computational time and prediction performance. Readers who are interested in data selection may refer to \cite{choi2023constraint}. 

We are now ready to present the Theorem for ensuring $\mathcal{C}$ \revise{is} forward invariant.
\begin{theorem} \label{the1}
    Given system \eqref{system} satisfying Assumptions 1 and 2, and a continuous differentiable function \eqref{func:h}, if the following condition holds, 
    \begin{align}
        &-\Gamma^TM(q)^{-1}(u-C(q,v)v-G(q)-\mu(x))-\gamma\frac{\partial \beta_1}{\partial z}\Gamma^Tv\nonumber\\&-v^T\left(\frac{\partial^2\eta }{\partial q^2}\right)^Tv\ge-\delta\beta_2(h(x))+\left\| \Gamma^TM(q)^{-1}\right\|\bar\lambda,\label{the1:condition}
    \end{align}
    then, the input $u$ renders $\mathcal{C}$ forward invariant.
\end{theorem}
\begin{proof}
    The proof of Theorem 1 is straightforward. \revise{We leverage GP regression to learn the unknown function $d_i(x),\forall i\in\mathbb{N}_n$, and utilize the deterministic upper bound $\bar\lambda$ of prediction error to compensate for the difference between the prediction mean and real value.} Taking the first derivative of \eqref{func:h} with respect to time and combining \eqref{prediction error} and \eqref{statistic bound}, we have
    \begin{align}
        \dot h(x)=&-\Gamma^TM(q)^{-1}(u-C(q,v)v-G(q)-\mu(x))\nonumber\\&-\Gamma^TM(q)^{-1}(\mu(x)-d(x))\nonumber\\&-\gamma\frac{\partial \beta_1}{\partial z}\Gamma^Tv-v^T\left(\frac{\partial^2\eta }{\partial q^2}\right)^Tv
        \nonumber\\\ge&-\Gamma^T(M(q)^{-1}(u-C(q,v)v-G(q)-\mu(x))\nonumber\\&-\gamma\frac{\partial \beta_1}{\partial z}\Gamma^Tv-\left\|\Gamma^T(M(q)^{-1}\right\|\bar\lambda\nonumber\\&-v^T\left(\frac{\partial^2\eta }{\partial q^2}\right)^Tv.
    \end{align}
    
    Then, according to \eqref{the1:condition}, $\dot h(x)\ge-\delta\beta_2(h(x))$ holds, which ensures $\mathcal{C}$ is forward invariant.
\end{proof}

\revise{The feasibility of \eqref{the1:condition} under actuator constraints is now dependent on the condition parameters $\gamma$ and $\delta$. Before determining precisely when it is enforced, we have to introduce one more assumption.}
\begin{assumption} \label{ass3}
    The first and second derivative of $f(q)$ and $g(q)$ with respect to $q$ are bounded, i.e., $\left\|\frac{\partial f_i}{\partial q}\right\| \le f_q,\ \left\|\frac{\partial g_i}{\partial q}\right\| \le g_q,\ \left\|\frac{\partial^2 f_i}{\partial q^2}\right\| \le f_{q^2}  \text{ and } \left\|\frac{\partial^2 g_i}{\partial q^2}\right\| \le g_{q^2}  $ for all $i=1,2,3$.
\end{assumption}

Note that Assumption 3 is a natural extension of Assumption 2, which is guaranteed by the continuity and smoothness of $f(q)$ and $g(q)$. Thus, Assumption 3 imposes no practical restrictions.
Suppose that Assumption 3 holds, we have the following Lemma.
\begin{lemma}
     The first and second derivative of $\eta(q)$ with respect to $q$ are bounded, and defined by: $\left\| \frac{\partial \eta(q)}{\partial q} \right\|\le \eta_{qmax}:=3(f_q+g_q)$ and $\left\| \frac{\partial^2 \eta(q)}{\partial q^2} \right\|\le \eta_{max^2}:=3\left( f_{q^2} +2f_qg_q+g_{q^2}\right)$
\end{lemma}
\begin{proof}
\revise{The proof can be found in Appendix A. }
\end{proof}



Now, we offer a sufficient condition \revise{that parameters $\gamma$ and $\delta$ should satisfy to ensure \eqref{the1:condition} holds}, as follows:
\begin{lemma}\label{lemma3}
    Given a system \eqref{system} with singularity constraints, velocity constraints and input constraints defined by \eqref{configuration constraint}, \eqref{velocity constraint} and \eqref{input constraint}. Suppose that there exists a function $h(x)$ defined by \eqref{func:h} with extended class-$\mathcal{K}$ functions $\beta_1$ and $\beta_2$. if the following condition holds for $q\in\mathcal{Z}\cap\mathcal{C}$ and $c\in\mathcal{C}\cap\mathcal{V}_i, \forall i\in\mathbb{N}_n$
    \begin{align}
        &3\eta_{max^2}v_{max}^2+\sqrt{3}\gamma\frac{\partial \beta_1}{\partial z}\eta_{max}v_{max}-\delta\beta_2(h(x))\nonumber\\&+\eta_{qmax}m_{max}(\sqrt{3}u_{max}+c_{max}v_{max}^2+g_{max}+\bar\lambda)\nonumber\\&+\eta_{qmax}m_{max}\left\|\mu(x)\right\|\le0, \label{sufficient condition}
    \end{align}
    then, there exists $u_i\in\mathcal{U}$ for all $i\in\mathbb{N}_n$ enforce \eqref{the1:condition}.
\end{lemma}
\begin{proof}
We first re-arrange \eqref{the1:condition} as follows:
\begin{align}
    &\Gamma^TM(q)^{-1}(u-C(q,v)v-G(q)-\mu(x))+\gamma\frac{\partial \beta_1}{\partial z}\Gamma^Tv\nonumber\\&+v^T\left(\frac{\partial^2\eta }{\partial q^2}\right)^Tv-\left\| \Gamma^TM(q)^{-1}\right\|\bar\lambda\le\delta\beta_2(h(x)). \label{new condi}
\end{align}

The proof of Lemma 3 starts with analyzing the upper bound of each term in the left part of \eqref{new condi}, since as long as the upper bound of the left part is smaller than $\delta\beta_2(h(x))$, \eqref{new condi} is guaranteed. First, the influence of control input satisfies
\begin{align}
    \Gamma^T M^{-1}u\le\left\|\Gamma^T M^{-1}u \right\|\le\sqrt{3}\eta_{qmax}m_{max}u_{max},
\end{align}
where $\left\|u\right\|\le\sqrt{3}u_{max}$ can be directly deduced by $u_i\in\mathcal{U}$.

Second, we address the term $-\Gamma M^{-1}Cv$. In light of Properties 1 and 2, it follows $\left\|M(q)^{-1}C(q,v) \right\|\le m_{max}c_{max}\left\|v\right\|$. Moreover, based on $v_i\in\mathcal{V}_i$, we have $\left\|v\right\|\le\sqrt{3}v_{max}$, and such that $\left\|M(q)^{-1}C(q,v)\right\|\le\sqrt{3}m_{max}c_{max}v_{max}$. Consequently, the following condition holds
\begin{align}
    -\Gamma M^{-1}Cv\le3\eta_{qmax}m_{max}c_{max}v_{max}^2. \label{C}
\end{align}

Third, we address the gravity term. Obviously, it follows
\begin{align}
    -\Gamma M^{-1}G\le\left\|\Gamma M^{-1}G\right\|\le\eta_{qmax}m_{max}g_{max}. \label{G}
\end{align}

Next, for the term $v^T\left(\frac{\partial^2\eta }{\partial q^2}\right)^Tv$, we have
\begin{align}
    v^T\left(\frac{\partial^2\eta }{\partial q^2}\right)^Tv\le3\eta_{max^2}v_{max}^2. \label{quadra}
\end{align}

Due to $\gamma\frac{\partial \beta_1}{\partial z}$ is non-negative, we also have
\begin{align}
    \gamma\frac{\partial \beta_1}{\partial z}\Gamma^Tv\le\sqrt{3}\gamma\frac{\partial \beta_1}{\partial z}\eta_{max}v_{max}. \label{beta1}
\end{align}

Finally, we address the terms related to GP regression, they follow
\begin{align}
    -\Gamma M^{-1}\mu(x)\le\eta_{qmax}m_{max}\left\|\mu(x)\right\|, \label{mean}
\end{align}
\begin{align}
    \left\|\Gamma^T(M(q)^{-1}\right\|\lambda(x)\le\eta_{qmax}m_{max}\bar\lambda. \label{variance}
\end{align}

Consequently, combining the above all terms, we have
\begin{align}
    \Phi:=&\Gamma^TM(q)^{-1}(u-C(q,v)v-G(q)-\mu(x))\nonumber\\&+\gamma\frac{\partial \beta_1}{\partial z}\Gamma^Tv+v^T\left(\frac{\partial^2\eta }{\partial q^2}\right)^Tv+\left\| \Gamma^TM(q)^{-1}\right\|\bar\lambda\nonumber\\&\le\eta_{qmax}m_{max}(\sqrt{3}u_{max}+3c_{max}v_{max}^2+g_{max})\nonumber\\&\quad\ \eta_{qmax}m_{max}(\left\|\mu(x)\right\|+\bar\lambda)+3\eta_{max^2}v_{max}^2\nonumber\\&\quad\ +\sqrt{3}\gamma\frac{\partial \beta_1}{\partial z}\eta_{max}v_{max} =:\Psi
\end{align}

Then, as long as \eqref{sufficient condition} holds, \revise{we can enforce \eqref{the1:condition} (and thus $\dot h(x)\ge-\delta\beta_2(h(x))$).}
\end{proof}
In addition, \eqref{sufficient condition} not only provides a sufficient condition to enforce forward invariance of \eqref{the1:condition}, but also offers a basis for the selection of parameters $\gamma$ and $\delta$. As $\beta_2(h(x))>0$ when $h(x)>0$, \eqref{sufficient condition} can be guaranteed if $\delta$ is chosen to be sufficiently large. Accordingly, we define the minimum value $\delta$ can reach as $\delta^{\ast}$
\begin{align}
    \delta^{\ast}:=\max_{q\in\mathcal{Z}\cap\mathcal{C},v_i\in\mathcal{C}\cap\mathcal{V}_i,i\in\mathbb{N}_n}\frac{\Psi}{\beta_2(h(x))}.
\end{align}

As long as we select $\delta\ge\delta^{\ast}$, the sufficient condition \eqref{sufficient condition} for Theorem 1  can be ensured. the above conclusion holds under the condition that $h(x)>0$. If $h(x)=0$, i.e., the boundary of the constraint $\mathcal{C}$ \revise{is reached}, $\delta$ is no longer effective in \eqref{sufficient condition}. In order to ensure the CBF condition \eqref{the1:condition} for singularity constraints still holds, parameter $\gamma$ should be tuned elaborately. To this end, we recall \eqref{the1:condition}.
\revise{Obviously, when $h(x)=0$, \eqref{the1:condition} is transformed to $\Phi\le0$. We keep $\Gamma^TM(q)^{-1}u$ unchanged in $\Phi$ and still use the upper bounds of the remaining terms \eqref{C} -\eqref{variance}, then we have}
\begin{align}
    \Phi\le&\Gamma^TM(q)^{-1}u+\eta_{qmax}m_{max}(3c_{max}v_{max}^2+g_{max})\nonumber\\&\eta_{qmax}m_{max}(\left\|\mu(x)\right\|+\lambda)+3\eta_{max^2}v_{max}^2\nonumber\\&+\sqrt{3}\gamma\frac{\partial \beta_1}{\partial z}\eta_{max}v_{max} \label{phi}
\end{align}

In order to ensure \eqref{the1:condition} holds, $\Phi\le0$ must be guaranteed. As $-\sqrt{3}\eta_{max}m_{max}u_{max}\le\Gamma^TM(q)^{-1}u\le\sqrt{3}\eta_{max}m_{max}u_{max}$, the limited actuation that system \eqref{system} can offer to ensure $\Phi\le0$ is $-\eta_{max}m_{max}u_{max}$. We make the following assumption to ensure system \eqref{system} has sufficient actuation capability.
\begin{assumption}
    The system has sufficient effort such that $u_{max}\ge\frac{\xi}{\sqrt{3}\eta_{qmax}m_{max}}$ with $\xi :=3\eta_{max^2}v_{max}^2+\eta_{qmax}m_{max}(c_{max}v_{max}^2+g_{max}+\left\|\mu(x)\right\|+\bar\lambda)$. 
\end{assumption}
Then, obviously if $\gamma$ is sufficiently small, $\Phi\le0$ can be guaranteed. we defined the maximum value of $\gamma$ as:
\begin{align}
    \gamma^{\ast}:=\min_{q\in\mathcal{Z}\cap\mathcal{C},v_i\in\mathcal{C}\cap\mathcal{V}_i,i\in\mathbb{N}_n} \frac{\sqrt{3}\eta_{qmax}m_{max}u_{max}-\xi}{\sqrt{3}\frac{\partial \beta_1}{\partial z}\eta_{max}v_{max}}.
\end{align}

As long as $\gamma\le\gamma^{\ast}$, $\Phi\le0$ can be guaranteed when $h(x)=0$, which implies that $\dot h(x)\ge0$ (as $\beta_2(0)=0$).
\subsection{Velocity constraints}
\revise{We address velocity constraints in this section, and} deduce the sufficient conditions to enforce \eqref{dbx}, we re-write \eqref{dbx} as follows:
\begin{subequations}
    \begin{align}
    e_i^TM^{-1}(u-Cv-G-d)\le k\beta_3(\overline{b}_i(v_i))\\
    e_i^TM^{-1}(u-Cv-G-d)\ge-k\beta_3(\underline{b}_i(v_i)),
    \end{align}
\end{subequations}
\begin{theorem}
    Given system \eqref{system} satisfying Assumption 1, and continuous differentiable functions \eqref{func:b}, if the following conditions hold for all $i\in\mathbb{N}_n$
    \begin{align}
        e_i^TM^{-1}(u-Cv-G-\mu)\le k\beta_3(\overline{b}_i)-\left\|e_i^TM^{-1}\right\|\bar\lambda \label{b condition 1}\\
        e_i^TM^{-1}(u-Cv-G-\mu)\ge-k\beta_3(\underline{b}_i)+\left\|e_i^TM^{-1}\right\|\bar\lambda, \label{b condition 2}
    \end{align}
    then, the input $u$ renders $\mathcal{V}_i$ forward invariant for all $i\in\mathbb{N}_n$.
\end{theorem}
\begin{proof}
    Similarly to the proof of Theorem 1, we differentiate $\overline{b}_i(v_i)$ and $\underline{b}_i(v_i)$ with respect to time, and combine them with \eqref{b condition 1} and \eqref{b condition 2}.
    \begin{align}
        \dot{\overline{b}}_i(v_i)=&-e_i^TM^{-1}(u-Cv-G-\mu+\mu-d)\nonumber\\
        \ge&-e_i^TM^{-1}(u-Cv-G-\mu)-\left\|e_i^TM^{-1}\right\|\bar\lambda\nonumber\\
        \ge&-k\beta_3(\overline{b}_i),
    \end{align}
    \begin{align}
        \dot{\underline{b}}_i(v_i)=&e_i^TM^{-1}(u-Cv-G-\mu+\mu-d)\nonumber\\
        \ge&e_i^TM^{-1}(u-Cv-G-\mu)-\left\|e_i^TM^{-1}\right\|\bar\lambda\nonumber\\
        \ge&-k\beta_3(\underline{b}_i).
    \end{align}
    Then, as long as \eqref{b condition 1} and \eqref{b condition 2} hold, the forward invariance of $\mathcal{V}_i$ for all $i\in\mathbb{N}_n$ can be guaranteed.
\end{proof}
Intuitively, with the increasing of $k$, it is easier to satisfy conditions \eqref{b condition 1} and \eqref{b condition 2} when $v_i$ is sufficiently far away from the boundary $v_{max}$ and $v_{min}$. 

Since conditions that parameters $\gamma$ and $\delta$ need to satisfy have been defined, i.e. $\gamma\le\gamma^{\ast}$ and $\delta\ge\delta^{\ast}$, the task of safety filter is to find a a control input which can simultaneously ensure that \eqref{the1:condition}, \eqref{b condition 1}, and \eqref{b condition 2} are valid. \revise{ The following assumption holds for the parameters $\gamma$, $\delta$ and $k$:}
\begin{assumption}
    As long as parameters satisfy $\gamma\le\gamma^{\ast}$, $\delta\ge\delta^{\ast}$ and $k\ge k^{\ast}$, there always exists control input $u_i\in\mathcal{U}$ for all $i\in\mathbb{N}_m$ that can simultaneously ensure \eqref{the1:condition}, \eqref{b condition 1} and \eqref{b condition 2} hold for all $i\in\mathbb{N}_n$.
\end{assumption}

\revise{Additional conditions to guarantee that there is no conflict between \eqref{b condition 1}, \eqref{b condition 2} might need to be enforced, which is subject of future research.}

\subsection{Main results}
\revise{Based on the above analysis regarding singularity and velocity constraints, we are now concluding the main results of this paper as the following theorem.}
\begin{theorem}
    Consider system \eqref{system} with singularity constraints, velocity constraints and actuator constraints defined by \eqref{configuration constraint}, \eqref{velocity constraint} and \eqref{input constraint}, and satisfying Assumptions 1, 2, 3, and 4. Let functions $h(x),\ \overline{b}_i(v(i))$ and $\underline{b}_i(v(i))$ be defined by \eqref{func:h} and \eqref{func:b} with three extended class-$\mathcal{K}$ functions. If the parameters satisfy $\gamma\le\gamma^{\ast}$ and $\delta\ge\delta^{\ast}$, then there always exists control input $u_i\in\mathcal{U}$ for all $i\in\mathbb{N}_m$ that can simultaneously ensure \eqref{the1:condition}, \eqref{b condition 1} and \eqref{b condition 2} hold for all $i\in\mathbb{N}_n$, i.e., $\exists u$ that renders sets $\mathcal{Z}\cap\mathcal{C}$ and $\mathcal{C}\cap\mathcal{V}_i, \forall i\in\mathbb{N}_n$ forward invariant.
\end{theorem}
\begin{proof}
    \revise{Based on Assumption 1, the unknown function $\left\|d_i(x)\right\|_{k_i}\le B_i,\forall i\in\mathbb{N}_n$. Then, according to Lemma 2, the prediction error of GP regression is bounded by a state-independent constant $\bar\lambda$. In lights of Lemma 3, as long as we choose $\gamma\le\gamma^{\ast}$ and $\delta\ge\delta^{\ast}$, there exists $u_i\in\mathcal{U},\ \forall i\in\mathbb{N}_n$ that can ensure \eqref{the1:condition} holds (and thus $\mathcal{C}$ is forward invariant). Moreover, Theorem 2 guarantees that there exists $u_i\in\mathcal{U},\forall i\in\mathbb{N}_n$ that can ensure $\mathcal{V}_i,\ \forall i\in\mathbb{N}_n$ is forward invariant by appropriately selecting parameter $k$. Assumption 5 ensures that there exists $u_i\in\mathcal{U},\ \forall i\in\mathbb{N}_n$ that can simultaneously enforce \eqref{the1:condition}, \eqref{b condition 1} and \eqref{b condition 2}. Accordingly, $\exists u$ renders sets $\mathcal{Z}\cap\mathcal{C}$ and $\mathcal{C}\cap\mathcal{V}_i, \forall i\in\mathbb{N}_n$ forward invariant.}
\end{proof}
\revise{Finally, we re-write the optimization problem \eqref{optimization} as follows:}
\begin{align}
    u^\ast = \text{arg}&\min_{u\in\mathcal{U}}\left\| u-u_{nom}\right\|^2 \nonumber\\
    &\ \ \text{s.t.}\ \eqref{the1:condition}, \eqref{b condition 1}\ \text{and}\ \eqref{b condition 2},\nonumber\\
    &\qquad\ u_i\in\mathcal{U},\forall i\in\mathbb{N}_n
\end{align}

\section{NUMERICAL VERIFICATION}
\begin{figure}[!t]
  \centering
  \includegraphics[width=0.8\linewidth]{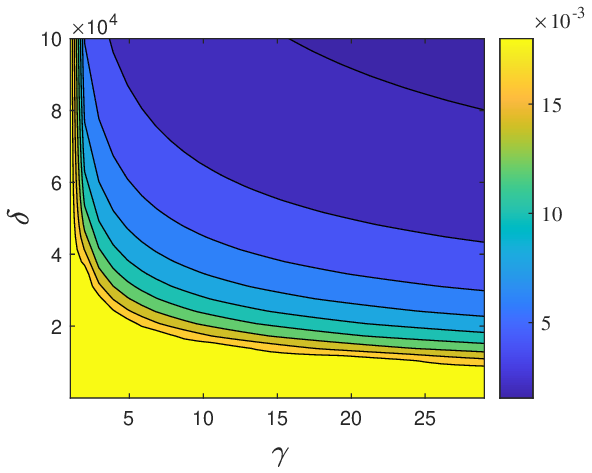}
  \caption{The impact of $\gamma$ and $\delta$ on the singularity constraint $z_{min}(q)$}
  \label{contourf}
\end{figure}
In this section, we employ a 2 DoFs planar manipulator \revise{(as shown in Fig. \ref{robots} (a))} with two identical links as example to validate the effectiveness of the proposed CBFs design approach for singularity avoidance. High-fidelity simulations are conducted on Simscape. The radius, length and density of links are set to be 0.01 m, 0.5 m, $7.8\times 10^3\text{kg}/\text{m}^3$, respectively. As it is a planar manipulator, $G(q)$ = 0. The actuation capacities of motors are set to be $u_{max}=-u_{min}=5$ Nm, while the velocity constraints are $v_{max}=-v_{min}=2$ rad/s. The hard constraints of joint angles are set to be $q_{max}=-q_{min}=\pi/3$. Three extended class-$\mathcal{K}$ functions are selected as: $\beta_1(z)=z$, $\beta_2(h)=h^3$ and $\beta_3(\overline{b}_i)=\tan^{-1}(\overline{b}_i)$ (i.e., $\beta_3(\underline{b}_i)=\tan^{-1}(\underline{b}_i)$). We attach a lumped mass at the end of second link with $m=0.2$ kg, and view it as unknown. GP regression with  squared-exponential kernel $k_i=sf^2\exp(\left\|x_1-x_2\right\|^2/el^2)$ with $sf = 0.01$, $el=1$ and $\sigma_v^2=0.001$ is leveraged to learn the unknown $d_i(x)$ for all $i\in\mathbb{N}_n$ . The size of dataset is chosen as $M=200$. Then, we can directly calculate $\bar\lambda=3.52$. The Jacobian matrix of the 2 DoFs manipulator is obtained as:
\begin{align}
    J=\begin{bmatrix}
-l_1\sin q_1-l_2\sin q_{12} &-l_2\sin q_{12}  \\
 l_1\cos q_1+l_2\cos q_{12}& l2\cos q_{12}
\end{bmatrix}
\end{align}
where $q_{12}=q_1+q_2+q_{ini}$. $l_1$ and $l_2$ denote the length of two links, respectively. $q_{ini}$ denotes an user-defined initial angle for the second link. Obviously, when these two links are parallel, $\det(J)=0$, i.e., the robot is under singular configuration. Accordingly, we define the unit direction vectors of two links as: $f(q)=[\cos(q_1),\ \sin(q_1),\ 0]^T$ and $g(q)=[\cos(q_1+q_2+q_{ini}),\ \sin(q_1+q_2+q_{ini}),\ 0]^T$.

In order to deduce the corresponding value of $\gamma^{\ast}$ and $\delta^{\ast}$, we first calculate $m_{max}$ and $c_{max}$. We can simply choose $m_{max}$ to be the maximum eigenvalue of $M(q)^{-1}$ for $q\in[q_{min},\ q_{max}]$, which is $m_{max}=49.246$. We can similarly calculate $c_{max}=0.243$. As $\beta_1(z)=z$, we have $\frac{\partial \beta_1}{\partial z}=1$. Then, We can directly calculate $\gamma^{\ast}=29.987$. We choose $\gamma$ slightly smaller than $\gamma^{\ast}$ to be $\gamma = 29$. $\delta^{\ast}=5.924$ is solved by Matlab function $fmincon$ where $q \in \mathcal{Z}\cap\mathcal{C}$ and $v_i\in\mathcal{C}\cap\mathcal{V}_i, \forall i\in\mathbb{N}_n$. We choose PID controller as the nominal controller with three parameters as: $k_p=200,\ k_i=200,\ k_v=10$. First of all, we analyze the impact of $\gamma$ and $\delta$ on the singularity constraint $z_{min}(q)$ \revise{where $z_{min}(q)$  is defined as the minimum $z(q)$ during the entire trajectory}, which reflects the degree of conservatism of CBFs. As shown in Fig. \ref{contourf}, as $\gamma$ and $\delta$ increase, the minimal value of $z(q)$ decreases, which implies that the control inputs drive the system states closer to the boundary of the singularity constraints $z(q)=0$ but never exceed it. Conversely, when $\gamma$ and $\delta$ are small, the system tends to move away from the constraint boundary. Therefore, increasing these parameters helps to reduce the conservatism of the CBFs.

By choosing $\gamma=29$ and $\delta=10^5$, the trajectory tracking results are shown in Fig. \ref{tracking}. In contrast to the results without compensation from GP regression, where neither singularity constraints nor velocity constraints can be satisfied due to the presence of model mismatch, all constraints are ensured by leveraging GP regression. Fig. \ref{z} illustrates a detailed change process of the singularity constraints, in which we can find that $z(q)\ge0$ is consistently guaranteed for GP based CBFs, whereas $z(q)<0$ (indicating a singularity configuration) occurs around 6.5 s without GP regression. The results also demonstrate the effectiveness of the proposed CBFs design methodology for singularity constraints subject to actuator constraints. \new{While the impact of CBFs on trajectory tracking accuracy is beyond the scope of this paper, it is noted that the system's tracking accuracy decreases near the constraint boundary due to the conservativeness introduced by CBFs. We will consider solutions to this issue in the future.}
\begin{figure}[!t]
  \centering
  \subfloat[]{\includegraphics[width=0.5\linewidth]{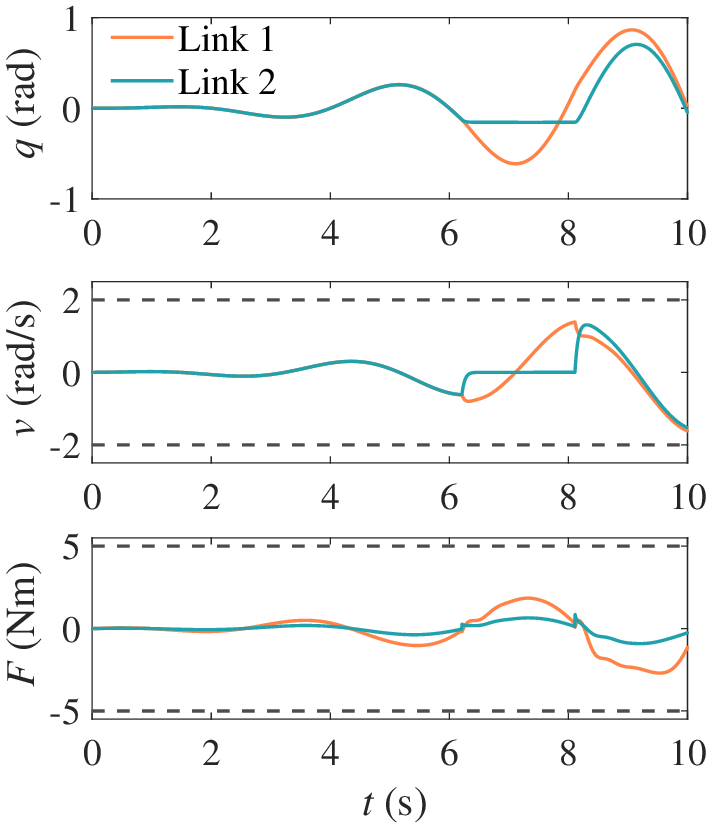}\label{}}
  \subfloat[]{\includegraphics[width=0.5\linewidth]{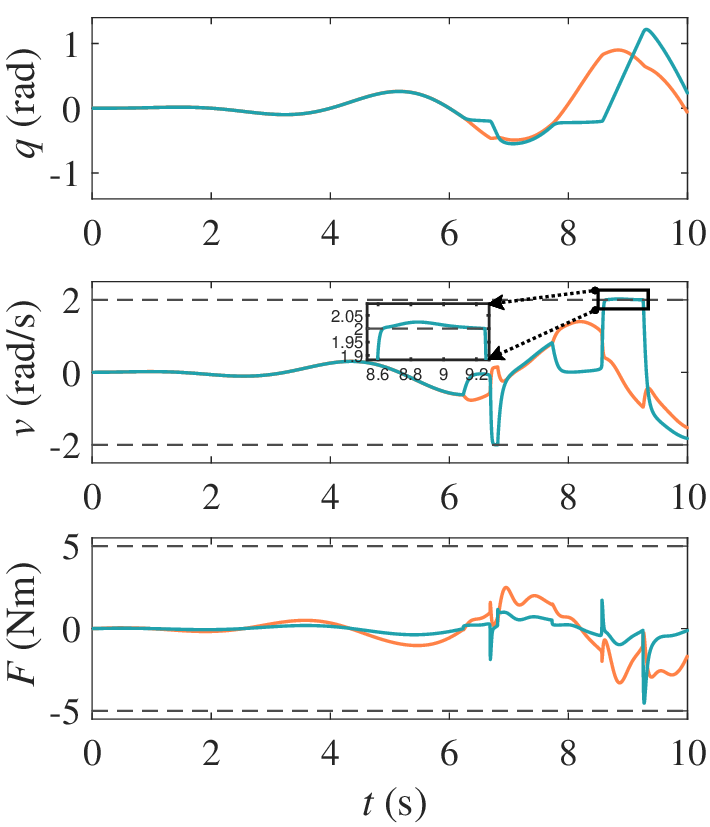}\label{}}
  \caption{Comparison of trajectory tracking results. (a) with GP regression. (b) without GP regression.}
  \label{tracking}
\end{figure}
\begin{figure}[!t]
  \centering
  \subfloat[]{\includegraphics[width=0.5\linewidth]{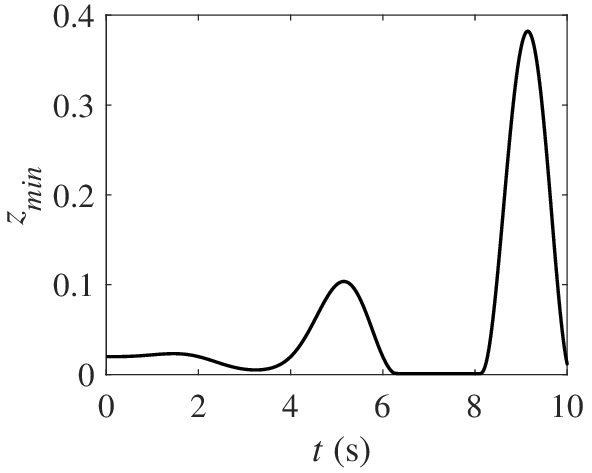}\label{Com:re}}
  \subfloat[]{\includegraphics[width=0.5\linewidth]{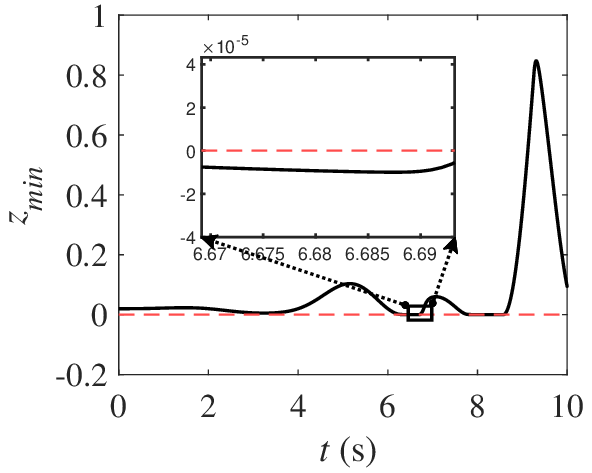}\label{Com:acc}}
  \caption{Comparison of singularity constraints. (a) with GP regression. (b) without GP regression.}
  \label{z}
\end{figure}
\section{CONCLUSION}
We proposed a methodology for CBFs design to address singularity avoidance problem in robotic systems subject to model mismatch and actuator constraints. The feasibility of CBFs under actuator constraints was guaranteed by the parameter selection criterion, and model mismatch was learned by GP regression. High-fidelity simulations validated the effectiveness of the proposed approach.

Several challenges remain to be addressed in the future. First, a universal method that encompasses all singularity configurations should be adopted. Additionally, the use of GP regression in this paper introduces a conservative condition for CBFs, which should be relaxed in future work.
{\appendix[]
\subsection{Appendix A}
The first derivative of $\eta$ is given as
    \begin{align}
    \left\| \frac{\partial \eta}{\partial q} \right\|=&\left\| \sum_{i=1}^{3}\left( \frac{\partial f_i(q)}{\partial q}g_i(q)+f_i(q)\frac{\partial g_i(q)}{\partial q} \right) \right\|\nonumber\\
    \le&\sum_{i=1}^{3}\left\| \frac{\partial f_i(q)}{\partial q}g_i(q) \right\|+\left\| f_i(q)\frac{\partial g_i(q)}{\partial q} \right\|\nonumber\\
    \le&3(f_q+g_q).
\end{align}

The second derivative of $\eta$ is deduced as
\begin{align}
     &\left\| \frac{\partial^2 \eta}{\partial q^2} \right\|=\left\| \sum_{i=1}^{3}\left( \frac{\partial^2 f_i}{\partial q^2}g_i+2\frac{\partial f_i}{\partial q}\left( \frac{\partial g_i}{\partial q} \right)^T+f_i\frac{\partial^2 g_i}{\partial q^2} \right) \right\|\nonumber\\
    &\le\sum_{i=1}^{3}\left( \left\| \frac{\partial^2 f_i}{\partial q^2}g_i \right\|+2\left\| \frac{\partial f_i}{\partial q}\left( \frac{\partial g_i}{\partial q} \right)^T \right\| +\left\| f_i\frac{\partial^2 g_i}{\partial q^2} \right\| \right)\nonumber\\
    &\le3\left( f_{q^2} +2f_qg_q+g_{q^2}\right).
\end{align}}
\bibliographystyle{IEEEtran}
\bibliography{myrefs}

\begin{thebibliography}{10}
\providecommand{\url}[1]{#1}
\csname url@samestyle\endcsname
\providecommand{\newblock}{\relax}
\providecommand{\bibinfo}[2]{#2}
\providecommand{\BIBentrySTDinterwordspacing}{\spaceskip=0pt\relax}
\providecommand{\BIBentryALTinterwordstretchfactor}{4}
\providecommand{\BIBentryALTinterwordspacing}{\spaceskip=\fontdimen2\font plus
\BIBentryALTinterwordstretchfactor\fontdimen3\font minus \fontdimen4\font\relax}
\providecommand{\BIBforeignlanguage}[2]{{%
\expandafter\ifx\csname l@#1\endcsname\relax
\typeout{** WARNING: IEEEtran.bst: No hyphenation pattern has been}%
\typeout{** loaded for the language `#1'. Using the pattern for}%
\typeout{** the default language instead.}%
\else
\language=\csname l@#1\endcsname
\fi
#2}}
\providecommand{\BIBdecl}{\relax}
\BIBdecl

\bibitem{li2023effective}
G.~Li, H.~Liu, T.~Huang, J.~Han, and J.~Xiao, ``An effective approach for non-singular trajectory generation of a 5-dof hybrid machining robot,'' \emph{Robotics and Computer-Integrated Manufacturing}, vol.~80, p. 102477, 2023.

\bibitem{pulloquinga2023type}
J.~L. Pulloquinga, R.~J. Escarabajal, {\'A}.~Valera, M.~Vall{\'e}s, and V.~Mata, ``A type ii singularity avoidance algorithm for parallel manipulators using output twist screws,'' \emph{Mechanism and Machine Theory}, vol. 183, p. 105282, 2023.

\bibitem{wabersich2023data}
K.~P. Wabersich, A.~J. Taylor, J.~J. Choi, K.~Sreenath, C.~J. Tomlin, A.~D. Ames, and M.~N. Zeilinger, ``Data-driven safety filters: Hamilton-jacobi reachability, control barrier functions, and predictive methods for uncertain systems,'' \emph{IEEE Control Systems Magazine}, vol.~43, no.~5, pp. 137--177, 2023.

\bibitem{katriniok2022control}
A.~Katriniok, ``Control-sharing control barrier functions for intersection automation under input constraints,'' in \emph{2022 European Control Conference (ECC)}.\hskip 1em plus 0.5em minus 0.4em\relax IEEE, 2022, pp. 1--7.

\bibitem{cortez2020correct}
W.~S. Cortez and D.~V. Dimarogonas, ``Correct-by-design control barrier functions for euler-lagrange systems with input constraints,'' in \emph{2020 American Control Conference (ACC)}.\hskip 1em plus 0.5em minus 0.4em\relax IEEE, 2020, pp. 950--955.

\bibitem{xiao2023barriernet}
W.~Xiao, T.-H. Wang, R.~Hasani, M.~Chahine, A.~Amini, X.~Li, and D.~Rus, ``Barriernet: Differentiable control barrier functions for learning of safe robot control,'' \emph{IEEE Transactions on Robotics}, vol.~39, no.~3, pp. 2289--2307, 2023.

\bibitem{sforni2024receding}
L.~Sforni, G.~Notarstefano, and A.~D. Ames, ``Receding horizon cbf-based multi-layer controllers for safe trajectory generation,'' in \emph{2024 American Control Conference (ACC)}.\hskip 1em plus 0.5em minus 0.4em\relax IEEE, 2024, pp. 4765--4770.

\bibitem{nguyen2021robust}
Q.~Nguyen and K.~Sreenath, ``Robust safety-critical control for dynamic robotics,'' \emph{IEEE Transactions on Automatic Control}, vol.~67, no.~3, pp. 1073--1088, 2021.

\bibitem{jagtap2020control}
P.~Jagtap, G.~J. Pappas, and M.~Zamani, ``Control barrier functions for unknown nonlinear systems using gaussian processes,'' in \emph{2020 59th IEEE Conference on Decision and Control (CDC)}.\hskip 1em plus 0.5em minus 0.4em\relax IEEE, 2020, pp. 3699--3704.

\bibitem{fisac2018general}
J.~F. Fisac, A.~K. Akametalu, M.~N. Zeilinger, S.~Kaynama, J.~Gillula, and C.~J. Tomlin, ``A general safety framework for learning-based control in uncertain robotic systems,'' \emph{IEEE Transactions on Automatic Control}, vol.~64, no.~7, pp. 2737--2752, 2018.

\bibitem{balta2021learning}
E.~C. Balta, K.~Barton, D.~M. Tilbury, A.~Rupenyan, and J.~Lygeros, ``Learning-based repetitive precision motion control with mismatch compensation,'' in \emph{2021 60th IEEE Conference on Decision and Control (CDC)}.\hskip 1em plus 0.5em minus 0.4em\relax IEEE, 2021, pp. 3605--3610.

\bibitem{kurtz2021control}
V.~Kurtz, P.~M. Wensing, and H.~Lin, ``Control barrier functions for singularity avoidance in passivity-based manipulator control,'' in \emph{2021 60th IEEE Conference on Decision and Control (CDC)}.\hskip 1em plus 0.5em minus 0.4em\relax IEEE, 2021, pp. 6125--6130.

\bibitem{hashimoto2022learning}
K.~Hashimoto, A.~Saoud, M.~Kishida, T.~Ushio, and D.~V. Dimarogonas, ``Learning-based symbolic abstractions for nonlinear control systems,'' \emph{Automatica}, vol. 146, p. 110646, 2022.

\bibitem{scharnhorst2022robust}
P.~Scharnhorst, E.~T. Maddalena, Y.~Jiang, and C.~N. Jones, ``Robust uncertainty bounds in reproducing kernel hilbert spaces: A convex optimization approach,'' \emph{IEEE Transactions on Automatic Control}, vol.~68, no.~5, pp. 2848--2861, 2022.

\bibitem{liu2012new}
X.-J. Liu, C.~Wu, and J.~Wang, ``{A New Approach for Singularity Analysis and Closeness Measurement to Singularities of Parallel Manipulators},'' \emph{Journal of Mechanisms and Robotics}, vol.~4, no.~4, p. 041001, 08 2012.

\bibitem{lin2016improving}
Z.~Lin, J.~Fu, H.~Shen, G.~Xu, and Y.~Sun, ``Improving machined surface texture in avoiding five-axis singularity with the acceptable-texture orientation region concept,'' \emph{International Journal of Machine Tools and Manufacture}, vol. 108, pp. 1--12, 2016.

\bibitem{khosravi2022safety}
M.~Khosravi, C.~K{\"o}nig, M.~Maier, R.~S. Smith, J.~Lygeros, and A.~Rupenyan, ``Safety-aware cascade controller tuning using constrained bayesian optimization,'' \emph{IEEE Transactions on Industrial Electronics}, vol.~70, no.~2, pp. 2128--2138, 2022.

\bibitem{choi2023constraint}
J.~J. Choi, F.~Castaneda, W.~Jung, B.~Zhang, C.~J. Tomlin, and K.~Sreenath, ``Constraint-guided online data selection for scalable data-driven safety filters in uncertain robotic systems,'' \emph{arXiv preprint arXiv:2311.13824}, 2023.

\end{thebibliography}

\end{document}